\documentclass[11pt,a4paper]{article}

\usepackage{amsmath, amssymb, amsthm}
\usepackage{geometry}
\usepackage{enumitem}
\usepackage{cite}
\usepackage{pgfplots}
\usepackage[hidelinks]{hyperref}
\newtheorem{theorem}{Theorem}
\newtheorem{conjecture}{Conjecture}
\newtheorem{lemma}[theorem]{Lemma}

\newtheorem{corollary}[theorem]{Corollary}
\theoremstyle{definition}

\theoremstyle{remark}

\DeclareMathOperator{\lcm}{lcm}

\title{The lonely runner conjecture holds for nine runners}
\author{Matthieu Rosenfeld\thanks{
 This research was funded, in whole or in part, by the French National Research Agency (ANR) under grant agreement No. ANR-24-CE48-3758-01. In accordance with the objective of open access dissemination, the author applies a Creative Commons Attribution (CC-BY) license to any accepted article or manuscript (AAM) resulting from this submission.
}\\
LIRMM, Univ Montpellier, CNRS, Montpellier, France
}
\date{\today}

\begin{document}

\maketitle
\begin{abstract}
We prove that the lonely runner conjecture holds for nine runners. Our proof is based on a couple of improvements of the method we used to prove the conjecture for eight runners.
\end{abstract}

\section{Introduction}
The lonely runner conjecture, first stated by Wills in 1967 \cite{phdWills1965,Wills1967Jun}, is a well-known open problem in combinatorial number theory with several applications and equivalent formulations \cite{Perarnau2025Nov}. The conjecture can be stated as follows.
\begin{conjecture}[Lonely Runner Conjecture]
For all integer $k \geq 1$, and every set of distinct integers $v_1, \ldots, v_{k}$, there exists a real number $t$ such that for all $i$, we have 
\[
  \|tv_i\| \geq \frac{1}{k+1}.
\]
\end{conjecture}
This conjecture is still wide open in general, but has been established for small values of $k$:
\begin{itemize}
\item $k=3$ the first non-trivial value \cite{Betke1972Jun,Cusick1973, Cusick1974Jan, Cusick1982Jan},
\item $k=4$ first resolved with computer aid~\cite{Cusick1984Oct} and later simplified~\cite{Bienia1998Jan},
\item $k=5$~\cite{Bohman2001Feb,Renault2004Oct},
\item $k=6$~\cite{Barajas2008Mar},
\item $k=7$ was recently solved with computer assistance by the author of the current article~\cite{Rosenfeld2025Sep}.
\end{itemize}
For a more detailed introduction to the lonely runner conjecture, we refer the reader to our previous article~\cite{Rosenfeld2025Sep} or the recent survey by Perarnau and Serra~\cite{Perarnau2025Nov}.
In this article, we adapt our technique from \cite{Rosenfeld2025Sep} to prove the case of nine runners. Note that a couple of days before the submission of this article as a preprint, Trakulthongchai independently announced a proof for the cases of nine and ten runners \cite{Tanupat2025Nov}. The result of \cite{Tanupat2025Nov} also build on the proof from \cite{Rosenfeld2025Sep}.

The article is organized as follows. Section 2 recalls the core argument from \cite{Rosenfeld2025Sep}. In Section 3, we present our main lemma and use it to prove the lonely runner conjecture for nine runners. Section 4 details the computer-assisted verification. The concluding section outlines possible improvements, discusses an attempted SAT-solver approach, and compares our results with those of \cite{Tanupat2025Nov}.

In the remainder, we say that a set $S\subseteq\mathbb{N}$ of size $k$ has the \emph{lonely runner property (abbreviated LR property)} if there exists $t$ such that for all $v\in S$, $\|tv\| \geq \frac{1}{k+1}.$

\section{The core of the argument from \cite{Rosenfeld2025Sep}}
The core of our argument relies on the following bound by Malikiosis, Santos and Schymura \cite{Malikiosis2024Nov}.
\begin{theorem}[{\cite[Theorem A]{Malikiosis2024Nov}}]\label{UpperBoundOnSum}
If the lonely runner conjecture holds for $k$, then it holds for $k+1$ for all $k$-uple $v_1, \ldots, v_k$ such that $\gcd(v_1, \ldots, v_k)=1$, and
\[
\sum_{S\subseteq \{1,\ldots, k\}} \gcd(\{v_i: i\in S\}) >\binom{k+1}{2}^{k-1}.
\]
\end{theorem}
In \cite{Rosenfeld2025Sep}, we used the following corollary to bound the product of the speeds in any minimal counter-example for $k$ runners.
\begin{corollary}[{\cite{Rosenfeld2025Sep}}]\label{cor:UpperBoundOnProduct}
If the lonely runner conjecture holds for $k$, then it holds for $k+1$ for all $k$-uple $v_1, \ldots, v_k$ such that $\gcd(v_1, \ldots, v_k)=1$, and
\[
\prod_{i=1}^{k}v_i\ge  \left[\frac{\binom{k+1}{2}^{k-1}}{k}\right]^k\,.
\]
\end{corollary}
Now the proof relies on finding many small prime divisors of the product of the speeds in any minimal counterexample for $k$ runners. With enough such divisors, we get a lower bound on the product that is larger than the upper bound given in Corollary \ref{cor:UpperBoundOnProduct}, implying that there is no such counterexample. The following lemmas from \cite{Rosenfeld2025Sep} allow to find such divisors.

\begin{lemma}[{\cite{Rosenfeld2025Sep}}]\label{lem:oldCovering}
Let $k\ge3$ be an integer such that the lonely runner conjecture holds for $k-1$. Let $p\in \mathbb{N}$ be a positive integer such that for all 
$v_1,\ldots, v_k\in \{0,\ldots, (k+1)p-1\}$ with
\begin{enumerate}[label=(\roman*)]
\item for all subset $S\in \{v_1,\ldots, v_k\}$ of size $k-1$, we have $\gcd(S\cup\{(k+1)p\})=1$,
\item for all $i$, $v_i$ is not divisible by $p$,
\end{enumerate}
there exists $t\in \{0,\ldots,(k+1)p-1\}$ such that for all $i$,
\[
\left\|\frac{tv_i}{(k+1)p}\right\|\ge \frac{1}{k+1}\,.
\]
Then for any set of $k$ distinct integers $\{v_1, \ldots, v_k\}$ that does not have the LR property, we have 
  \[
 p \text{ divides } \prod_{i=1}^{k}v_i.
  \]
\end{lemma}

For each prime $p$, a finite computation can be used to verify the condition of this lemma. Since many primes happen to satisfy this condition when $k=7$, we obtain a list of many primes that must divide the product of the speeds of a minimal counterexample. With enough such primes, one can contradict the upper bound on the product of the speeds given by Corollary \ref{cor:UpperBoundOnProduct} and prove the lonely runner conjecture for $k+1$ runners.
One ingredient of the proof of Lemma \ref{lem:oldCovering} is the following lemma (the same idea was already used in \cite{Bienia1998Jan}).
\begin{lemma}[{\cite{Rosenfeld2025Sep}}]\label{StrongerGCDCondition}
Let $k \geq 3$ be an integer such that the lonely runner conjecture holds for $k-1$. Let $v_1, \ldots, v_k$ be a set of integers with $\gcd(\{v_1, \ldots, v_k\})=1$ but $\gcd(v_1,\ldots, v_{k-1})\not=1$. Then $\{v_1, \ldots, v_k\}$ has the LR property.
\end{lemma}

For the current article, we will use a stronger version of Lemma \ref{lem:oldCovering} given in Lemma \ref{lem:needOfCovering}. For this, we provide a version of Lemma \ref{StrongerGCDCondition} tailored for the case $k=8$ with more general conditions. Hence, Lemma \ref{lem:needOfCovering} is also tailored to the case $k=8$ which allows a weaker condition (i.e., we have fewer $k$-uples to verify). At the same time, we also generalize a bit the conditions of the lemma by adding a coefficient $c$ that allows to apply our lemma to powers of small primes as well.

\section{Proof of the lonely runner conjecture for nine runners}
We now provide our version of Lemma \ref{StrongerGCDCondition} tailored for $k=8$. The proof is similar to the proof of Lemma \ref{StrongerGCDCondition} given in \cite{Rosenfeld2025Sep}.

\begin{lemma}\label{EvenStrongerGCDCondition}
Let $k \geq 3$ be an integer such that the lonely runner conjecture holds for all $k'<k$. Let $v_1, \ldots, v_8$ be a set of integers with $\gcd(\{v_1, \ldots, v_8\})=1$. Then
\begin{itemize}
\item if at least $6$ of the $v_i$ are divisible by $3$, then $\{v_1, \ldots, v_8\}$ have the LR property,
\item if at least $5$ of the $v_i$ are divisible by $9$, then $\{v_1, \ldots, v_8\}$ have the LR property.
\end{itemize}
\end{lemma}
\begin{proof}
We start with the first item, and we assume that at least $6$ of the $v_i$ are divisible by $3$. By the $\gcd$ condition at least one of the $v_i$ is not divisible by $3$ and by Lemma \ref{StrongerGCDCondition}, if exactly one of them is not divisible by $3$ then the set has the LR property, so the only remaining case is when exactly $6$ of the $v_i$ are divisible by $3$. By reordering the $v_i$, we can assume that $v_1,\ldots, v_6$ are divisible by $3$ and $v_7$ and $v_8$ are not divisible by $3$.
Let $t$ be a lonely runner solution for $\{v_1,\ldots,v_{6}\}$, that is, for all $i\in \{1,\ldots, 6\}$, 
\[
\left\|v_i\cdot t\right\|\ge \frac{1}{7}>\frac{1}{9}\,.
\]
Hence, for all $t'\in \{t, t+1/3, t+2/3\}$ and for all $i\in \{1,\ldots, 6\},$
\[
\left\|v_i\cdot t'\right\|=\left\|v_i\cdot t\right\|\ge \frac{1}{9}\,,
\]
where the first equality holds since $v_i$ is divisible by $3$.
On the other hand, for $i\in \{7,8\}$, $\gcd(v_i,3)=1$, and we have 
\[
\left\{v_i\cdot (t+j/3) \mod 1: j\in\{0,1,2\}\right\}
=\left\{v_i\cdot t +j/3 \mod 1: j\in\{0,1,2\}\right\}\,.
\]

For each $i\in \{7,8\}$, this set corresponds to three equally spaced points on the unit circle, and at most one of these points satisfies $\left\|v_i\cdot (t+j/3)\right\|<1/9$ (the forbidden interval has length $2/9<1/3$). Therefore, there exists at least one $t'\in\{t,t+1/3,t+2/3\}$ for which $\left\|v_i\cdot t'\right\|\ge 1/9$ holds for both $i=7,8$, and thus for all $i=1,\dots,8$ as desired.

The argument for the second item is similar. From the first item, we only need to take care of the case where exactly $5$ of the $v_i$ are divisible by $9$, and the others are not divisible by $3$. That is, up to reordering, $v_1,\ldots, v_5$ are divisible by $9$, and $v_6,v_7,v_8$ are not divisible by $3$. As previously, we let $t$ be a lonely runner solution for $v_1,\ldots,v_{5}$, and we consider the set $S=\{t, t+1/9,\ldots, t+8/9\}$. For any $t'\in S$ and any $i\in\{1,2,3,4,5\}$, $||v_it'||= ||v_it||\ge 1/6>1/9$. On the other hand, for $i\in \{6,7,8\}$, $\gcd(v_i,9)=1$, and we have
\[
\left\{v_i\cdot (t+j/9) \mod 1: j\in\{0,\ldots,8\}\right\}
=\left\{v_i\cdot t +j/9 \mod 1: j\in\{0,\ldots,8\}\right\}\,.
\]
For each $i\in \{6,7,8\}$, this set corresponds to nine equally spaced points on the unit circle, and the forbidden interval is an open interval of size $2/9$ so at most two such points are forbidden for each $i$. Hence, at least $3$ of the choices of $t'\in S$ are such that $\left\|v_i\cdot t'\right\|\ge \frac{1}{9}$ for all $i\in \{1,\ldots,8\}$, as desired.
\end{proof}

We are now ready to prove our main lemma.
\begin{lemma}\label{lem:needOfCovering}
Suppose that $d,c\in \mathbb{N}$ are positive integers such that for all 
$v_1,\ldots, v_8\in \{0,\ldots, 9cd-1\}$ with
\begin{enumerate}[label=(\roman*)]
\item for all subset $S\in \{v_1,\ldots, v_8\}$ of size $7$, we have $\gcd(S\cup\{9cd\})=1$,
\item at most $5$ of the $v_i$ are divisible by $3$,
\item at most $4$ of the $v_i$ are divisible by $9$,
\item $\prod_{i=1}^8 v_i$ is not divisible by $d$,
\end{enumerate} 
there exists $t\in \{0,\ldots,9cd-1\}$ such that for all $i$,
\[
\left\|\frac{tv_i}{9cd}\right\|\ge \frac{1}{9}\,.
\]
Then for any set of $8$ distinct integers $\{w_1, \ldots, w_8\}$ that does not have the LR property, we have 
  \[
 d \text{ divides } \prod_{i=1}^{8}w_i.
  \]
\end{lemma}
\begin{proof}
Let $d,c$ be as in the theorem statement and let $\{w_1, \ldots, w_8\}$ be a set of $8$ distinct integers without the LR property. Without loss of generality, we can assume $\gcd(\{w_1, \ldots, w_8\})=1$, otherwise we can divide every $w_i$ by $\gcd(w_1, \ldots, w_8)$ as this new set does not have the LR property either and the product of the speed divides the previous product.

For all $i$, let $v_i$ be the remainder of the euclidean division of $w_i$ by $9cd$. 
Since the set $\{w_1, \ldots, w_8\}$ does not have the LR property, we have by Lemma \ref{StrongerGCDCondition} and Lemma \ref{EvenStrongerGCDCondition} that
\begin{itemize}
\item for all subset $S\in \{w_1,\ldots, w_8\}$ of size $7$, we have $\gcd(S)=1$, so in particular $\gcd(S\cup\{9cd\})=1$,
\item at most $5$ of the $w_i$ are divisible by $3$,
\item at most $4$ of the $w_i$ are divisible by $9$,
\end{itemize}
Since each $v_i$ is obtained by reducing $w_i$ modulo $9cd$, the same properties hold for the $v_i$, that is we have (i), (ii) and (iii) from our Theorem hypothesis.

Now, suppose for the sake of contradiction that $\prod_{i=1}^8 w_i$ is not divisible by $d$, then $\prod_{i=1}^8 v_i$ is not divisible by $d$ either which implies (iv).
By our theorem hypothesis there exists $t\in \{0,\ldots,9cd-1\}$ such that for all $i$,
\[
\left\|\frac{tv_i}{9cd}\right\|\ge \frac{1}{9}\,.
\]
For all $i$, $w_i\equiv v_i \mod 9cd$ implies,
\[
\left\|\frac{tw_i}{9cd}\right\|=\left\|\frac{tv_i}{9cd}\right\|\ge \frac{1}{9}\,,
\]
which contradicts the fact that  $\{w_1, \ldots, w_8\}$ does not have the LR property. Hence, $d$ divides $\prod_{i=1}^8 w_i$, as desired.
\end{proof}
Using this lemma, and a computer program that verifies its conditions, we can find many integers $d$ that divide the product of the speeds in any minimal counterexample for $9$ runners. In particular, let $S$ be the set given by
\begin{align*}
  S=\{&64,81,25,121,169,17,19,23,29,31,37,41,43,47,53,59,61,67,71,73,79,83,89,97,101,\\
  &103,107,109,113,127,131,137,139,149,151,157,163,167,173,179,181,191\}\,, 
\end{align*}
and for all $d\in S$, let $c_d$ be such that\footnote{These values are the smallest values of $c$ for which our program was able to verify the conditions of Lemma \ref{lem:needOfCovering}.}
\[
c_d=
\begin{cases}
5 & \text{if } d=25\,,\\
3 & \text{if } d=17\text{ or }19\,,\\
2 & \text{if } d\in \{64,23,29,41\}\,,\\
1 & \text{otherwise.}
\end{cases}
\]
The program provided with this article \cite{gitRepo} can verify that for all $d\in S$, the pair $(d,c_d)$ satisfies the conditions of Lemma \ref{lem:needOfCovering}. Note that $S$ contains a power of every prime number not larger than $191$ other than $7$.\footnote{It is not hard to prove that $7$ must divide the product, but it does not help. On the other hand with $p=49$, if there is a $c_p$ that respects the conditions of our lemma it is at least $5$ (as we whecked for $\{1,2,3,4\}$) and the verification is would be more expensive than the one for $191$.} This provides enough divisors to reach the bound given by Corollary \ref{cor:UpperBoundOnProduct} and prove the lonely runner conjecture for $9$ runners.
\begin{proof}[Proof of the lonely runner conjecture for nine runners]
For the sake of contradiction suppose there exists a set $\{v_1,\ldots,v_8\}$ of $8$ distinct integers without the LR property. Without loss of generality, we may assume  $\gcd(\{v_1,\ldots,v_8\})=1$. We let 
\[
\mathcal{P}=\prod_{i=1}^8 v_i.
\]

By Corollary \ref{cor:UpperBoundOnProduct}, we have
\begin{equation}\label{eq:product8}
\mathcal{P} \le \left[\frac{\binom{9}{2}^{7}}{8}\right]^8\approx 8.47657\times 10^{79} < 9\times 10^{79}\,.
\end{equation}

Our implementation of Lemma \ref{lem:needOfCovering}, can be used to verify that $\mathcal{P}$ is divisible by $d$, for all $d\in S$ where
\begin{align*}
  S=\{&64,81,25,121,169,17,19,23,29,31,37,41,43,47,53,59,61,67,71,73,79,83,89,97,101,\\
  &103,107,109,113,127,131,137,139,149,151,157,163,167,173,179,181,191\}\,.
\end{align*}
Note that the elements of $S$ are pairwise co-prime (as each is a power of a distinct prime).
 Hence, $\mathcal{P}$ is divisible by 
\[
 \lcm(S)=\prod_{d\in S} d \approx 9.09778\times 10^{79}> 9\times 10^{79}\,.
\]
Since $\mathcal{P}$ is a positive integer this contradicts \eqref{eq:product8}. Therefore, there does not exist such a set $\{v_1,\ldots,v_8\}$ of $8$ distinct integers without the LR property. This concludes the proof of the lonely runner conjecture for nine runners.
\end{proof}

\section{Implementation details}\label{sec:implementation}
Following the approach of \cite{Rosenfeld2025Sep}, before discussing the implementation details, we first reformulate the condition of Lemma \ref{lem:needOfCovering} in a formulation closer to set-cover.

Given integers $c$ and $d$, we say that $v\in \{1,\ldots,9cd-1\}$ \emph{covers} $j\in\mathbb{N}$ if 
$\left\|\frac{jv}{9cd}\right\|< \frac{1}{8}$. We then say that a set $\{v_1,\ldots,v_i\}$ covers a set $S\subseteq\mathbb{N}$ if any $s\in S$ is covered by at least one of the $v_i$. Intuitively, if $\{v_1,\ldots,v_i\}$ does not cover $j$, then the time $\frac{j}{9cd}$ is a suitable time to verify the LR property for $\{v_1,\ldots,v_i\}$.

For all positive integers $i$, $j$, and $t$, we have $\left\|\frac{ij}{t}\right\|=\left\|\frac{(t-i)j}{t}\right\|$. This means that if $i$ covers $j$, then $i$ covers $9cd-j$, and $9cd-i$ covers $j$ and $9cd-j$.
Using this notion of cover and these symmetries\footnote{Another obvious symmetry that we do not exploit is that if $i$ covers $j$, then $j$ covers $i$.}, Lemma \ref{lem:needOfCovering} can be reformulated equivalently as follows.

\begin{lemma}\label{lem:reformulated}
Let $c,d\in \mathbb{N}$ be positive integers.
Suppose that there exists no $8$-uple $v_1,\ldots, v_8\in \{1,\ldots, \lfloor9cd/2\rfloor\}$ with
\begin{enumerate}[label=(\roman*)]
\item for all subset $S\in \{v_1,\ldots, v_8\}$ of size $7$, we have $\gcd(S\cup\{9cd\})=1$,
\item at most $5$ of the $v_i$ are divisible by $3$,
\item at most $4$ of the $v_i$ are divisible by $9$,
\item $\prod_{i=1}^8 v_i$ is not divisible by $d$,
\item $v_1,\ldots, v_k$ covers $\{1,\ldots, \lfloor9cd/2\rfloor\}$.
\end{enumerate} 
Then the conclusion of Lemma \ref{lem:needOfCovering} holds, that is,  for any set of $8$ distinct integers $\{v_1, \ldots, v_8\}$ without the LR property, we have 
  \[
 d \text{ divides } \prod_{i=1}^8 v_i.
  \]
\end{lemma}
That is, we need to verify that there exists a set cover (property $(v)$) with some extra divisibility constraints (properties $(i)$ to $(iv)$). Our program implements basic backtracking techniques to exhaustively search for a "bad" cover for a selected pair $c$ and $d$. The implementation is similar to the one used in \cite{Rosenfeld2025Sep} for $7$ runners, but a few extra tricks improve the running time significantly. 

First, for each $i\in \{1,\ldots, \lfloor(k+1)p/2\rfloor\}$ the set of positions covered by $i$ is precomputed and stored as vectors of bits (\verb!bitsets! in C++), which provides really efficient bitwise operations. 

Then we construct a cover element by element, and for each partial cover $\{v_1,\ldots,v_i\}$, before adding a new element we verify some necessary conditions for the partial cover to be extendable to a cover. If these conditions are not met, we know that we can backtrack, and otherwise we try to add a new element to the cover. Our conditions are based on computing how many new integers can be covered by the best choice of integers among the remaining available integers. More precisely, we choose one of the uncovered integers $u$ to be special, and for every integer $x$ that covers $u$, we compute the largest number $s$ of new integers that can be covered by adding some integers $x'$ to the partial cover $\{v_1,\ldots,v_i,x\}$. If this number $s$ multiplied by the number of integers that can be added to the partial cover (that is $8-i-1$) is smaller than the number of integers that are not covered yet by $\{v_1,\ldots,v_i,x\}$, then there is no way to complete the cover to a full cover satisfying property $(v)$ of Lemma \ref{lem:reformulated}, and we know that $x$ should not be added to the cover\footnote{Choosing $x$ and then $X'$, might seem overly complicated compared to simply choosing $x'$ directly, but this allows for slightly more efficient partial pre-computations.}. This pruning strategy is the main difference with the implementation of \cite{Rosenfeld2025Sep}, and it drastically improves the running time. 

Since the order of the elements in the cover does not matter, given a partial solution $\{v_1,\ldots,v_i\}$, we can always require $v_{i+1}$ to cover a given integer $u$. In particular, we can take $u$ to be the integer that can be covered by the fewest integers (thus minimizing the number of possible choices for $v_{i+1}$ before considering pruning).

Finally, whenever we finish the backtracking over a partial cover $\{v_1,\ldots,v_i, v_{i+1}\}$, we remember for all future completions of $\{v_1,\ldots,v_i\}$, that there is no reason to try to add $v_{i+1}$ (since we already know that no valid cover contains $\{v_1,\ldots,v_i, v_{i+1}\}$), so we eliminate it from the remaining available integers for the completions of the partial cover $\{v_1,\ldots,v_i\}$.

The implementation is available in the git repository linked with this article \cite{gitRepo}. This implementation is experimentally much faster than the one used in \cite{Rosenfeld2025Sep} even for $7$ runners. The run time for $k=7$, $d=163$ and $c=1$ which was the worst case and took $32$ hours is reduced to $50$ minutes on the same machine. Moreover, Lemma \ref{lem:needOfCovering} means, we would only need to go to smaller prime than $163$ for the case of $7$ runners, making this proof even faster as well. 

We provide in figure \ref{fig:runtime} the running time of our implementation for various values of $d\cdot c$. The running time seems to grow roughly as $O((d\cdot c)^{5.7})$ (the curve $x^{5.7}/(5\times10^7)$ is also plotted for comparison). It might not be significant, but we note that the few points that are far below the line are all cases where $d\cdot c$ has at least two divisors ($51,81,125,128$). We also provide in the appendix the Table \ref{tab:runtimes} that shows the running times for each $d$ used in our proof.
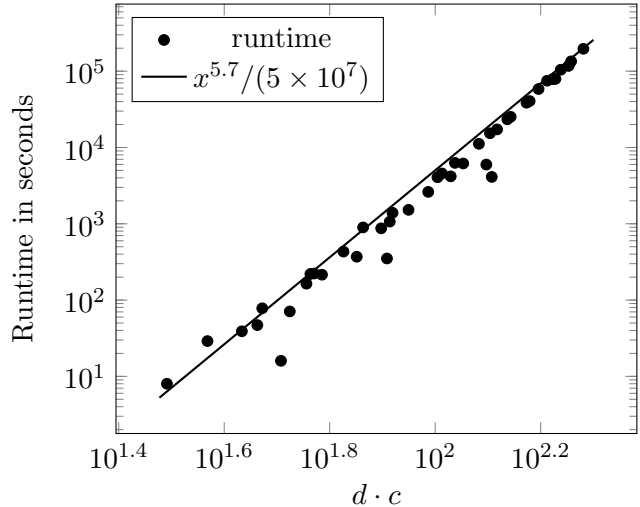
\begin{figure}[t!]
  \centering
\begin{tikzpicture}
\begin{axis}[
    xmode=log,
    ymode=log,
    xlabel={$d\cdot c$},
    ylabel={Runtime in seconds},
    legend pos=north west
]

\addplot[
    only marks,
    mark=*,
]
table {runtime.data};
\addlegendentry{runtime}

\addplot[
    thick,
    domain=30:200,
    samples=200
]
{ x^(5.7)/(5*10^7) };
\addlegendentry{$x^{5.7}/(5\times10^7)$}

\end{axis}
\end{tikzpicture}

\caption{Runtime of our implementation for various values of $d\cdot c$, and the curve $x^{5.7}/(5\times10^7)$.}
\label{fig:runtime}
\end{figure}

\section{Conclusion} 
As suggested in \cite{Rosenfeld2025Sep}, it required small optimizations to push the technique used for $8$ runners to solve the $9$ runners case. As our technique can be applied for every number of runners from $3$ to $9$, it seems reasonable to expect that it could theoretically be used for any particular number of runners. In practice, the computational cost grows fast with the number of runners, and with the size of the primes that need to be considered.

\paragraph{SAT-solver--based ones implementations}

In \cite{Rosenfeld2025Sep}, we raised the idea of using state-of-the-art
solvers instead of our ad-hoc implementation to improve the running time. Based on the results of the 2025 PACE contest, the best approaches to solve set cover instances all seem to rely directly on SAT-solvers~\cite{BibEntry2025Dec}. Since our computation is really close to a set cover problem, it seems natural to try to encode our problem as a SAT instance and use a state-of-the-art SAT solver to verify the condition of Lemma \ref{lem:reformulated}. The transformation of set-cover to SAT is pretty straightforward, and the other constraints we have concern the cardinality of some sets, so we use a SAT solver that allows cardinality constraints.

We provide with our implementation a program that generates such an encoding for any choice of $k$ and $d$\footnote{Note that our implementation does not allow to choose the value of $c$ and is equivalent to using $c=1$.}. This program generates an encoding to verify the condition of Lemma \ref{lem:reformulated} with Cardinality-CDCL \cite{Cardinality-CDCL} or with Minicard \cite{minicard} (which are respectively cardinality-constraints extensions of the SAT solvers CaDiCaL \cite{cadical} and MiniSat \cite{minisat}). Running any of these two solvers is experimentally much slower than our ad-hoc implementation. For instance, with $k=5$ and $d=31$, our implementation takes roughly $0.02$ seconds, while Cardinality-CDCL and Minicard take respectively $220$ seconds and $59$ seconds to conclude that the formula is unsatisfiable.

David Silver \cite{Silver2025LonelyRunnerSAT} also implemented the SAT-solver approach to verify the condition of Lemma \ref{lem:reformulated} using Kissat \cite{BiereFallerFazekasFleuryFroleyksPollitt-SAT-Competition-2024-solvers}. He reaches the same conclusion that the SAT-solver approach is not competitive with our ad-hoc implementation. We do not claim to be able to provide a theoretical explanation of this phenomenon, but our problem is quite different from "typical" SAT instances for which these solvers are optimized (our instances are monotone SAT, with very large clauses that behave in a "pseudo-random" manner, together with extra cardinality constraints).

These SAT-solver--based implementations can still be used to verify some of the smallest values of $d$ and $k$.

\paragraph{Towards $10$ runners.}
The bound given by Corollary \ref{cor:UpperBoundOnProduct} for 10 runners is roughly $2.77408\times 10^{110}$, which means that we would need to consider many more primes (or prime powers) than for $9$ runners, and larger primes as well. One would need to be able to verify the conditions of (a variant of) Lemma \ref{lem:needOfCovering} for primes up to roughly $263$ (depending slightly on the prime powers used for small primes). It is relatively simple to adapt the lemma and our code for the $10$ runner case. Running our implementation for the primes between $37$ and $71$ we can try to interpolate the running time for larger primes (see Fig. \ref{fig:runtime10}). Based on this rough estimation, it seems that verifying the condition for $p=263$ would take roughly two years on a single processor core with our current implementation (and we also need to run this for all smaller primes). It is certainly doable (and is easily parallelizable), but probably not worth the effort.
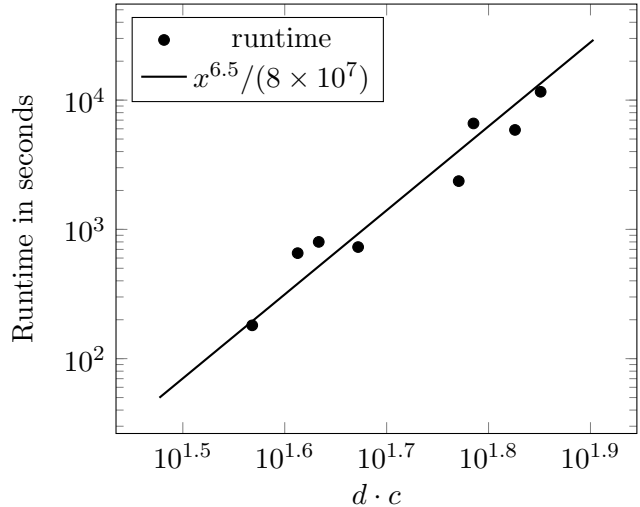
\begin{figure}[t!]
  \centering
\begin{tikzpicture}
\begin{axis}[
    xmode=log,
    ymode=log,
    xlabel={$d\cdot c$},
    ylabel={Runtime in seconds},
    legend pos=north west
]

\addplot[
    only marks,
    mark=*,
]
table {runtime_10.data};
\addlegendentry{runtime}

\addplot[
    thick,
    domain=30:80,
    samples=200
]
{ x^(6.5)/(8*10^7)};
\addlegendentry{$x^{6.5}/(8\times10^7)$}

\end{axis}
\end{tikzpicture}

\caption{Runtime of our implementation for $k=9$ and $c=1$ for various values of $d$ along with the curve $x^{6.5}/(8\times10^7)$. Based on the similar plot for $k=8$ (Figure \ref{fig:runtime}), it seems to be a reasonable estimation of the running time for larger values of $d$.}
\label{fig:runtime10}
\end{figure}

Some extra ideas are needed to make this computation more reasonable. 
This is what was done by Trakulthongchai in \cite{Tanupat2025Nov} to solve the case of $10$ runners.
One of the improvement here and in \cite{Tanupat2025Nov} is to generalize Lemma \ref{lem:oldCovering} (from \cite{Rosenfeld2025Sep}) by replacing $k+1$ by any multiple $c(k+1)$ of $k+1$  (see Lemma \ref{lem:needOfCovering}). The other main improvements that we use here were the following: 
\begin{itemize}
\item considering powers of small primes instead of only primes, which allow us to consider fewer and smaller primes in general,
\item finding some extra divisibility conditions based on Lemma \ref{EvenStrongerGCDCondition} that help to reduce the number of cases to verify,
\item improving the implementation used in \cite{Rosenfeld2025Sep} as discussed in Section \ref{sec:implementation}.
\end{itemize} 
In comparison, \cite{Tanupat2025Nov} uses a clever sieving approach to make the verification much faster (based on what is reported in \cite{Tanupat2025Nov} they do seem to be faster by a factor of order $100$). Some of the improvements that we used could be combined with the sieving approach of \cite{Tanupat2025Nov} to further improve the verification time (in particular, being able to consider power of small primes).

\bibliography{biblio}{}
\bibliographystyle{abbrv}
\appendix
\clearpage
\section{Table of the running times}
\begin{table}[h]
\centering
\footnotesize
\begin{tabular}{|c|c|c|c|}
\hline
$p$ & $c$ & Running time (seconds) \\\hline\hline
17        &3         &16       \\\hline
19        &3         &164      \\\hline
23        &2         &47       \\\hline
25        &5         &5974     \\\hline
29        &2         &220      \\\hline
31        &1         &8        \\\hline
37        &1         &29       \\\hline
41        &2         &1067     \\\hline
43        &1         &39       \\\hline
47        &1         &78       \\\hline
53        &1         &71       \\\hline
59        &1         &223      \\\hline
61        &1         &215      \\\hline
64        &2         &4115     \\\hline
67        &1         &432      \\\hline
71        &1         &370      \\\hline
73        &1         &896      \\\hline
79        &1         &872      \\\hline
81        &1         &351      \\\hline
83        &1         &1394     \\\hline
89        &1         &1523     \\\hline
97        &1         &2617     \\\hline
101       &1         &4071     \\\hline
103       &1         &4566     \\\hline
107       &1         &4185     \\\hline
109       &1         &6284     \\\hline
113       &1         &6174     \\\hline
121       &1         &11170    \\\hline
127       &1         &15367    \\\hline
131       &1         &17263    \\\hline
137       &1         &23409    \\\hline
139       &1         &25232    \\\hline
149       &1         &38621    \\\hline
151       &1         &40566    \\\hline
157       &1         &58429    \\\hline
163       &1         &75059    \\\hline
167       &1         &78611    \\\hline
169       &1         &79241    \\\hline
173       &1         &104496   \\\hline
179       &1         &117250   \\\hline
181       &1         &134716   \\\hline
191       &1         &196571   \\\hline
\end{tabular}
\caption{Table of the running times for each pair $(p,c)$ used to prove the lonely runner conjecture for $9$ runners with the implementation available at the git repository \cite{gitRepo}.}
\label{tab:runtimes}
\end{table}
\end{document}